\documentclass[11pt]{article}
\usepackage{graphicx} 

\usepackage{amsmath}
\usepackage{amssymb}

\usepackage{pre}
\usepackage[colorinlistoftodos,prependcaption,textsize=tiny]{todonotes}

\title{Shor's algorithm requires Fanout}

\def\ANON{0} 
\def\COMM{0} 

\ifnum\ANON=1
\author{
  Anonymous Authors.
}
\else
\author{
  Lucas Gretta \thanks{University of California at Berkeley. \ Email: \url{lucas_gretta@berkeley.edu}. \ Supported by NSF Award CCF-2231095.}
  \and
  Malvika Raj Joshi\thanks{University of California at Berkeley. \ Email: \url{malvika@berkeley.edu}. \ Supported by NSF grant 2311733 and DOE grant DE-SC0024124.}}
\fi
\date{}

\ifnum\COMM=1
\newcommand{\malvika}[1]{\textcolor{purple}{[\textbf{Malvika:} {#1}]}}
\newcommand{\luke}[1]{\textcolor{blue}{[\textbf{Luke:} {#1}]}}
\newcommand{\meghal}[1]{\textcolor{red}{[\textbf{Meghal:} {#1}]}}
\newcommand{\gpt}[1]{\textcolor{orange}{[\textbf{GPT:} {#1}]}}
\newcommand{\todomal}[1]{\todo[linecolor=Plum,backgroundcolor=Plum!25,bordercolor=Plum]{\textbf{@mal todo}: #1}}
\newcommand{\todoluke}[1]{\todo[linecolor=blue,backgroundcolor=blue!25,bordercolor=blue]{\textbf{@luke todo}:#1}}
\newcommand{\todomeghal}[1]{\todo[linecolor=red,backgroundcolor=red!25,bordercolor=red]{\textbf{@meghal todo} #1}}
\newcommand{\todoany}[1]{\todo[linecolor=orange,backgroundcolor=orange!25,bordercolor=orange]{\textbf{todo @plz do}: #1}}
\else
\newcommand{\malvika}[1]{}
\newcommand{\luke}[1]{}
\newcommand{\meghal}[1]{}
\newcommand{\gpt}[1]{}
\newcommand{\todomal}[1]{}
\newcommand{\todoluke}[1]{}
\newcommand{\todomeghal}[1]{}
\newcommand{\todoany}[1]{}
\fi

\begin{document}

\maketitle
\ifnum\ANON=1
\PackageWarningNoLine{Global}{Note anonymous mode}
\ifnum\COMM=1
\textcolor{red}{\textbf{Warning:}{ Comments are enabled in anonymous mode}}
\fi
\fi
\ifnum\COMM=1
\PackageWarningNoLine{Global}{Note comments are enabled}
\fi

\begin{abstract}
Shor's algorithm is a canonical quantum supremacy target whose core operation relies on the Quantum Fourier Transform (QFT). In this note, we answer the open question of \cite{fang2006qlb} by proving that approximating QFT in constant depth, for any $n$-qubit modulus, necessarily requires the $n$-qubit $\FANOUT$ operation. 

Formally, let $\QFT_q$ be the gate acting on $n = \lceil \log q \rceil$ qubits that computes the QFT under modulus $q$. It is known that any $n$-qubit $\QFT_q$ can be implemented in constant depth using $\FANOUT_n$, i.e. $\QFT_q \in \QAC^0_f$ \cite {hoyer2005fanout}. We prove the converse by using a $\QFT_q$ gate to construct a state of ``non-negligible felinity" \cite{gretta2026paritynotinqac0iff}. Consequently, $\QFT_q \in \QACZ \iff \FANOUT_n \in \QACZ$.  

In the case of $q = 2^n$, such as in Shor's, we approximate $\FANOUT_n$ using a single $\QFT_{2^n}$ gate and $O(1)$ two-qubit local gates, thus tying the feasibility of realizing Shor's  algorithm with NISQ circuits to that of $\FANOUT$.

\end{abstract}

\section{Introduction}
The Quantum Fourier Transform ($\QFT_q$) requires $\Theta(\log n)$ depth using only local gates \cite{qncqftlb}. In 2002, Høyer and Špalek showed that $\QFT_q$ can be performed in $O(1)$ depth if this gateset is extended to include $\FANOUT$ \cite{hoyer2005fanout}. 
However, it is yet unknown if $\QFT_q$ is implementable by potentially weaker nonlocal interactions. In particular, whether $\QFT_q \in \QACZ$ is open. We show this question is equivalent to Moore's question of whether $\FANOUT_n \in \QACZ$ \cite{moore1999qac0}. This affirmatively resolves the open question posed by Fang, Fenner, Green, Homer and Zhang in 2006: ``Is fanout really necessary to do the quantum
Fourier transform in constant depth?'' \cite{fang2006qlb}. 

For additional context on $\QACZ$ and $\QACZF$, we refer the reader to \cite{green2002acc, rosenthal2021qac0, grier2026tc0} and \cite{gretta2026paritynotinqac0iff}. 

\begin{theorem}[QFT implies Fanout]\label{thm:qftfan} 
For any $q \in \mathbb{N}$ and $n = \lceil \log_2 q \rceil$, if there exists a $\QACZ$ circuit $C$ approximating $\QFT_q$ to a non-negligible fidelity $\delta$, then $\FANOUT_n \in \QACZ$.
\end{theorem}
\noindent It follows from known results that $\FANOUT_n \in \QACZ \Rightarrow \QFT_q \in \QACZ$. Although Høyer–Špalek was approximate as stated, it can be made exact by combining with \cite{mosca2003exactquantumfouriertransforms} (\cref{sec:ampamp}). Thus,  as corollary, $\QFT_q \in \QACZ \iff \FANOUT_n \in \QACZ$.

\paragraph{Proof Overview.}
The key is to prepare a non-negligible ``felinity'' state using $C$, which by \cite{gretta2026paritynotinqac0iff} implies exact $\FANOUT_n$. 
The proof starts with a $n+1$-qubit product state whose $\QFT_q$ ($2^n \leq q < 2^{n+1}$) has constant amplitude on two adjacent frequencies with complementary binary representations. For instance, the frequencies $2^{n}-1$ and $2^{n}$ have binary representations $01\dots 1$ and $10\dots 0$ respectively. Thus, the $\QFT_q$ of such a state has constant felinity. To use $C$ in place of $\QFT_q$, we first mark the aforementioned branches to isolate them, allowing a $\Omega(\delta^2)$  felinity bound. 

Shor's algorithm \cite{Shor_1997} specifically uses a power-of-two modulus, i.e., a $\QFT^{\dag}_{2^n}$\footnote{or $\QFT_{2^n}$ depending on convention. The convention doesn't materially affect our proof.} gate. Although the above reduction applies, it is technically under $\QACZ$, not $\QNC^0$. In \cref{lem:localfanshor}, we describe a circuit for approximating $\FANOUT_n$ to arbitrary (depth-dependent) precision using only a $\QFT^{\dagger}_{2^n}$, two-qubit controlled-phases and no extra ancillae.
Our circuit  effectively approximates a controlled version of Høyer–Špalek's decrement gate \cite{hoyer2005fanout}. We truncate the resulting controlled-single-qubit-unitaries to the most significant qubits to obtain $O(1)$ depth. This lemma ties the realization of Shor's algorithm in the NISQ era (e.g. \cite{skosana2021demonstration}) to that of $\FANOUT$.
\section{The reductions}
Following \cite{hoyer2005fanout}, define $\omega_q := e^{2\pi i/q}$ and for $x \in \clr{0, \dots q-1}$,
\begin{align}
    \QFT_q \ket{x} := \frac{1}{\sqrt{q}} \sum_{y = 0}^{q-1} \omega_q^{xy} \ket{y}.
\end{align}
Denote $\aqftd$ as any gate approximating $\QFT_q$ to $\delta$ fidelity, i.e, for all normalized $\ket{\psi} \in \mathrm{span} \clr{\ket{0}, \ket{1} \dots \ket{q-1}}$, $\vlr{\braket{\psi | \QFT^\dag_q \cdot \aqftd |\psi}}^2 \geq \delta$. Note this notion captures a much weaker approximation than the typical $(1-\eps)$ fidelity. 

To represent such a \emph{qudit} $\ket{x}$ with $m \geq \lceil \log_2 q \rceil$ \emph{qubits}, we use the binary representation for $x$ with the leftmost qubit being the MSB. For instance, $\ket{2^{m-1}} = \ket{1}\ket{0}^{\otimes m-1}$. 

We define $\FANOUT_n$ as a unitary map satisfying $\ket{0}_t \ket{0}_A^{\otimes n} \mapsto \ket{0} \ket{0}_A^{\otimes n}$ and $\ket{1}_t \ket{0}_A^{\otimes n} \mapsto \ket{1} \ket{1}_A^{\otimes n}$. Some works refer to a generalized map that applies a $\cnot$ from $t$ to each qubit in $A$. These are closely related, as for any $U$ approximating $\FANOUT_n$, a corresponding approximation for generalized fanout exists using $U,U^\dag$ and local gates \cite{rosenthal2021qac0}. $\FANOUT_n \in \QACZ$ also implies $\FANOUT_{\poly(n)} \in \QACZ$. 

We will utilize the following Dirichlet kernel identity. For $D_q(L) := \sum_{r = 0}^{L-1} \omega_q^{r}$,
\begin{align}\label{eq:dirichlet}
    D_q(L) &= e^{i\pi (L-1)/q} \frac{\sin \frac{L \pi}{q}}{\sin \frac{\pi}{q}}.
\end{align}

Finally, we will apply the result of \cite{gretta2026paritynotinqac0iff}: any $\QACZ$ circuit preparing a $n$-qubit state $\rho$ of felinity $\geq n^{-c}$ for a constant $c$, implies $\FANOUT_n \in \QACZ$. The felinity, $\feln(\rho)$ is defined as,
 \begin{align}\label{eq:felin}
     \feln(\rho) := 2 \sum_{\y \in \bin^n} \braket{\y|\rho|\y} \cdot \braket{\y | X^{\tens n} \rho X^{\tens n} | \y}.
 \end{align}

First we show the exact case. 
\begin{lemma}[$\QFT_q$ produces high felinity]\label{lem:exactcase}
Let $N := 2^n$, $n \geq 1$, $N \leq q < 2N$ and $y_0, y_1 \in \clr{0, \dots q - 1}$ such that $y_0 \equiv N \mod q$ and $y_1 \equiv N-1 \mod q$.  
Then, there is a product state $\ket{\psi_q}$ satisfying, $\vlr{\braket{y_0 | \QFT_q | \psi_q} }\geq 1/2$ and $\vlr{\braket{y_1 | \QFT_q | \psi_q}} \geq 1/\pi$.
\end{lemma}
\begin{proof}
The claimed state is, 
\begin{align} 
\ket{\psi_q} &:= \ket{00} \bigotimes_{j = 2}^{n} \frac{\ket{0}+\omega_q^{-N2^{n-j}} \ket{1}}{\sqrt{2}} = \frac{1}{\sqrt{N/2}} \sum_{x = 0}^{N/2-1} \omega_q^{-Nx}  \ket{x}.
\end{align}
Note the leading $0$s ensure the state is  only supported on the qudits $\clr{0, \dots N/2-1}$ rather than the entire domain of $\QFT_q$, $\clr{0, \dots q-1}$. This gives,
\begin{align}\label{eq:here}
\QFT_q \ket{\psi_q} &= \sqrt{\frac{2}{qN}} \sum_{y = 0}^{q-1} \sum_{x = 0}^{N/2-1} \omega_q^{(y-N)x} \ket{y}.  
\end{align}
\noindent Immediately, $\vlr{\braket{y_0 | \QFT_q | \psi_q}} = \sqrt{\frac{N/2}{q}} \geq \frac{1}{2}$ because $q \leq 2N$. Also, since $1/2 \geq N/ (2q) \geq 1/4$, 
\begin{align}
    \vlr{\braket{y_1 | \QFT_q | \psi_q}} &= \vlr{\sqrt{\frac{2}{qN}} \sum_{x=0}^{N/2-1} \omega^{-x}_q} \\
    &=  \vlr{\sqrt{\frac{2}{qN}} \cdot D_{q}(N/2)} \\
    &= \sqrt{\frac{2}{qN}}  \cdot \vlr{\frac{\sin \frac{\pi N}{2q}}{\sin \frac{\pi}{q}}} \\
        &\geq \sqrt{\frac{2}{qN}} \cdot \frac{q}{\sqrt{2}\pi}  \geq \frac{1}{\pi}.
\end{align}
\end{proof}
\noindent We will also need the following lemma.
\begin{lemma}[Approx. Fourier states $\in \QACZ$]\label{lem:approx}
For $q$ sufficiently large and $n$ such that $2^n \leq q < 2^{n+1}$, and $y \in \clr{0, \dots q-1}$, define $\ket{\phi_y} := \QFT_q^\dag \ket{y}$. For any $\gamma \leq 1$ such that $\gamma \geq n^{-c}$ for some constant $c$, there exists $\ket{\phi'_y}$ preparable in $\QACZ$ satisfying, 
  $\vlr{\braket{\phi'_y | \phi_y}}^2 \geq 1-\gamma$.
\end{lemma}
\begin{proof}
The state $\ket{\phi_y} = \frac{1}{\sqrt{q}} \sum_{x=0}^{q-1} \omega_q^{-xy} \ket{x}$ when $q = 2^n$ is simply a product state. In fact, for any integer $a$, any $m$-qubit state of the following form is a product state.  
\begin{align}
    \ket{\eta_m(a)} &:= 2^{-m/2} \sum_{x = 0}^{2^m-1} \omega_q^{ax} \ket{x} 
\end{align}
When $2^n < q < 2^{n+1}$, first observe for any $\ell = O(\log n)$-qubit register and integer $2^{\ell-1} \leq R \leq 2^{\ell}$, the following state is preparable in $\QACZ$ ($O(1)$ depth, $\poly(n)$ ancillae). 
\begin{align}
  \ket{\zeta_{\ell}(R,a)} &:= \frac{1}{\sqrt{R}} \sum_{x = 0}^{R-1} \omega_q^{ax} \ket{x}  
\end{align}
This is because the comparator $\ket{x} \ket{0} \mapsto \ket{x} \ket{\blr{x < R}}$ has a size-$O(\ell)$ $\ACZ$ circuit and is thus implementable in $\QACZ$  \cite{grier2026tc0}. Applying it to $\ket{\eta_{\ell}(a)} \ket{0}$, one obtains,
\begin{align}
    \ket{\zeta'_{\ell}(R,a)} :=  \sqrt{\frac{R}{2^{\ell}}} \cdot \ket{\zeta_{\ell}(R,a)} \ket{1} + \sqrt{1-\frac{R}{2^{\ell}}} \ket{\bad} \ket{0}
\end{align}
which, since $R/2^{\ell} \geq 1/2$, can be amplified in $O(1)$ depth to $\ket{\zeta_{\ell}(R,a)}$ (see Cor. 3.16  \cite{gretta2026polylogarithmicweightdickestatesqac0}). 

Next, set $\ell :=  \lceil \log_2 \frac{2}{\gamma} \rceil  = O(\log n)$ and let $k := n-\ell+1$. 
Let $R$ be the integer formed by the $\ell$ most significant bits of $q$, formally $R := \lfloor q/ 2^k \rfloor$. Since $q \geq 2^n$, $2^{\ell-1} \leq R < 2^{\ell}$. Then, the desired state is, 
   \begin{align}
       \ket{\phi'_y} &:=  \ket{\zeta_{\ell}(R,-2^ky)} \otimes \ket{\eta_k(-y)} \\
       &= \frac{1}{\sqrt{2^kR}} \sum_{z = 0}^{R-1} \sum_{r = 0}^{2^k-1} \omega_q^{-2^kzy-ry} \ket{2^kz + r} \\
       &= \frac{1}{\sqrt{q'}} \sum_{x = 0}^{q'-1} \omega_q^{-xy} \ket{x}
   \end{align}
where $q' := R \cdot 2^{k}$. Indeed, this state is preparable in $\QACZ$ and  $\vlr{\braket{\phi_y | \phi'_y}}^2 = \frac{q'}{q} \geq 1-\frac{2^{k}}{q} \geq 1-\gamma$.
\end{proof}

\paragraph{Proof of \cref{thm:qftfan}.}
Parameterizing as $2^n \leq q < 2^{n+1}$, suppose that $C$ implements $\aqftd$ for $\delta \geq n^{-O(1)}$. We prepare a state $\ket{\nu_q}$ in $\QACZ$ such that $\aqftd \ket{\nu_q}$, which is preparable with $C$, has non-negligible felinity on $n$ qubits.
Set $y_0, y_1$ as in \cref{lem:exactcase}, thus, $\ket{y_1} = \ket{0}\ket{1}^{\otimes n}$, and $\ket{y_0} = \ket{1}\ket{0}^{\otimes n}$ for $q > 2^n$ and $\ket{0}^{\otimes n+1}$ otherwise. If $C$ is exact, $C(\ket{\psi_q})$ already has the required felinity, otherwise proceed as below.

For any state $\ket{\varphi}$ preparable in $\QACZ$, it is known that the reflection about it, $I - 2\kb{\varphi}$, is also implementable in $\QACZ$ (Fact 3.9 \cite{gretta2026polylogarithmicweightdickestatesqac0}). 

Let $A$ be an $n+1$ qubit register and $w_1, w_2$ be two additional qubits. 
Setting $\gamma := 0.001\delta$, for each $b \in \bin$, define the following unitaries using the state from \cref{lem:approx} and the product state $\ket{\+}_{w_1} \ket{b}_{w_2}$.
\begin{align}
    R_b &:= I - 2 \kb{\phi_{y_b}}_A \otimes \kb{\+b}_{w_1, w_2} \\
    R'_b &:= I - 2 \kb{\phi'_{y_b}}_A \otimes \kb{\+b}_{w_1, w_2} 
\end{align}
Note the operator norm $\Vlr{R_b - R'_b} = 2\sqrt{1-\vlr{\braket{\phi'_{y_b} | \phi_{y_b}}}^2} \leq 2\sqrt{\gamma} \leq 0.1\sqrt{\delta}$,  $R'_b$ is implementable in $\QACZ$, and $R'_0, R'_1$ commute because they reflect about orthogonal states. 

Then, for $\ket{\epr}_{w_1,w_2} = \frac{\ket{00}+\ket{11}}{\sqrt{2}}$, one can also prepare $\ket{\nu_q} := R'_0 R'_1 \ket{\psi_q}_A \ket{\epr}_{w_1,w_2}$ in $\QACZ$. It now suffices to argue that the state $\ket{\psi'} := (\aqftd)_A \otimes I_{w_1,w_2} \ket{\nu_q}$ has non-negligible felinity. Below we suppress the $I_{w_1,w_2}$ and $\aqftd$ implicitly acts only on $A$. 

Let $\ket{z_b} := \ket{y_b}_A \ket{(1-b)}_{w_1} \ket{b}_{w_2}$. In binary representation, $z_0,z_1$ are complementary on $\geq n+1$ bits. Observe $R'_{1-b}$ acts as identity on any $\ket{\psi}_{A,w_1} \ket{b}_{w_2}$ and $\ket{1-b}_{w_1}\ket{b}_{w_2}$ is orthogonal to $\ket{\epr}$ for both values of $b$. 
Putting it all together,  
\begin{align}
     \vlr{\braket{z_b | \aqftd | \nu_q}} &= \vlr{\bra{z_b} \cdot \aqftd \cdot  R'_b \ket{\psi_q}_A \ket{\epr}_{w_1,w_2}} \\
     &\geq \vlr{\bra{z_b} \aqftd \cdot R_b \cdot \ket{\psi_q} \ket{\epr}} - \underbrace{\Vlr{R_b - R'_b}}_{\leq 0.1 \sqrt{\delta}} \\
     &\geq \vlr{2 \bra{z_b} (\aqftd)_A \cdot \kb{\phi_{y_b}}_A \kb{\+b}_{w_1,w_2} \cdot \ket{\psi_q}_A \ket{\epr}_{w_1,w_2}} - 0.1\sqrt{\delta} \\
     &= \vlr{2 \braket{y_b| \aqftd | \phi_{y_b}} \braket{\phi_{y_b} | \psi_q} \braket{(1-b)b|\+b} \braket{\+b|\epr}} - 0.1\sqrt{\delta} \\
     &\geq \vlr{\braket{y_b | \aqftd \cdot \QFT_q^\dag | y_b}} \cdot 1/(\pi\sqrt{2})  - 0.1\sqrt{\delta} \\
     &\geq  0.11 \sqrt{\delta} 
\end{align}
Dropping $w_1, w_2$ and the leftmost qubit $a_0$, $\feln(\tr_{a_0w_1w_2} \kb{\psi'}) \geq 4 \vlr{\braket{z_0 | \psi'}}^2 \vlr{\braket{z_1 | \psi'}}^2 = \Omega(\delta^2)$, which is non-negligible. From \cite{gretta2026paritynotinqac0iff}, preparing $\ket{\psi'}$ in $\QACZ$ implies $\FANOUT_n \in \QACZ$. This concludes the proof.

\subsection{Local reduction for the Shor modulus.}
For a single-qubit unitary $U$, $\ctrl{U}(a,b)$ denotes the controlled $U$ operation with qubit $a$ as the control and $b$ as the target. Then, we show the following depth $k+2$ unitary approximates fanout. 

\begin{lemma}[Shor Case]\label{lem:localfanshor}
Let $A$ be a set of $n$ qubits representing a qudit register, labeled $a_1 , a_2 \dots a_n$  left to right and let $t$ be a input register and let $N = 2^n$. Define $P_s := \mathrm{diag}(1,e^{-2\pi i/ 2^s})$ and
$$U_f^{\id{k}} := \QFT^\dag_{N}(A)  \lr{\Pi_{s \in [k]} \ctrl{P_{s}}(t,a_s) }  H^{\tens n}_A$$
Then, for $1 \leq k \leq n$, $U_f\id{k} \ket{0}_t \ket{0}^{\otimes n}_A = \ket{0}_t \ket{0}^{\otimes n}_A$ and $ U_f\id{k} \ket{1}_t \ket{0}^{\otimes n}_A$ has at least $1-4^{1-k}$ fidelity with $\ket{1}_t \ket{1}^{\otimes n}_A$. 
\end{lemma}
\begin{proof}
Observe that the actions of $H^{\tens n}$ and $\QFT_N$ are identical on $\ket{0}^{\otimes n}$. When the input is $\ket{0}$, the map collapses to $\QFT^\dag_N H^{\tens n}_A \ket{0}^{\otimes n} = \QFT^\dag_N \QFT_N \ket{0}^{\otimes n} = \ket{0}^{\otimes n}$, as intended. 

Define $\ket{\eta_k}_{A} :=  \bigotimes_{s \in [k]} P_{s}(a_s) \cdot \ket{+}^{\tens n}_A$. 
Letting  $M := 2^{n-k}$, 
\begin{align}
    \ket{\eta_k} &= \lr{\bigotimes_{s \in [k]} \frac{\ket{0} + \omega^{-2^{n-s}}_N \ket{1}}{\sqrt{2}}}_{a_1\dots a_k} \otimes \ket{+}^{\otimes n-k}_{a_{k+1} \dots a_n} \\
    &= 2^{-n/2} \sum_{j = 0}^{2^k-1} \sum_{r = 0}^{M-1} \omega^{-Mj}_N \ket{Mj + r}_A.
\end{align}
Observe that $\ket{1}^{\otimes n}_A = \ket{N-1}_A$. Therefore, when the input is $\ket{1}$, 
\begin{align}
\bra{1}_t \bra{N-1}_A  U_f\id{k} \ket{1}_t \ket{0}_{A} &=  
\bra{N-1}_A \cdot \lr{\bra{1}_t \cdot U_f\id{k} \cdot \ket{1}_t} \ket{0^n}_A  \\
&= \braket{N-1 | \QFT^{\dag}_N | \eta_k} \\
&= 2^{-n} \sum_{j = 0}^{2^k-1} \sum_{r = 0}^{M-1} \omega_N^{Mj+r} \omega_N^{-Mj} \\
&= 2^{-n+k} \sum_{r = 0}^{M-1} \omega_N^{r} \\
&= \frac{1}{M} D_N(M)
\end{align}
Using \cref{eq:dirichlet} and third order taylor approximation of $\sin M \theta$ with $\theta = \pi/N$, 
\begin{align}
   \vlr{\bra{1}_t \bra{N-1}_A  U_f\id{k} \ket{1}_t \ket{0}_{A}} &\geq \frac{1}{M \theta} \lr{M\theta - \frac{M^3\theta^3}{6}}\\ 
   &\geq 1 - \frac{M^2 \pi^2}{6\cdot N^2} \\
   &\geq 1- \frac{\pi^2}{6\cdot 4^k}
\end{align}
Therefore, the fidelity is $\geq 1- \frac{\pi^2}{3 \cdot 4^k} \geq 1- 4^{1-k}$ as claimed.
\end{proof}

\ifnum\ANON=1
\else
\fi

\section{Acknowledgements.} 
The authors are grateful to Meghal Gupta for enlightening discussions about $\QACZ$ and to Angelos Pelecanos for his wisdom. 
The authors also acknowledge OpenAI’s ChatGPT (GPT-5.6 Pro) for helpful exchanges.
\appendix
\bibliographystyle{alpha}
\bibliography{main}
\section{Exact QFT using Fanout}\label[appendix]{sec:ampamp}
We provide more details on exact version of $\QFT_q \in \QACZF$, which is known folklore.  

First note that $\QACZF = \QNC^0_f = \mathsf{QTC}^0$ \cite{hoyer2005fanout, takahashi2012collapse}, and thus the $\TC^0$ arithmetic in Høyer–Špalek can be performed exactly. For power-of-two moduli, the only remaining error is in the phase estimation part. This is exactly what \cite{mosca2003exactquantumfouriertransforms} addresses by making this error uniform and removing it with exact amplitude amplification. 

For non-power-of-two moduli, the missing piece is an exact preparation of $\ket{\phi_0} := \frac{1}{\sqrt{2}} \sum_{x = 0}^{q-1} \ket{x}$. One can prepare $\ket{\phi_0}$ in $\QACZF$ as follows: (1) Mark the $x \le q$ branches of $\ket{+}^{\otimes \lceil \log_2 q \rceil}$ using $n$-bit comparator (available in  $\AC0 / \TC0$), (2) perform exact amplitude amplification to get $\ket{\phi_0}$ in $O(1)$ depth. This is same as the first part of \cref{lem:approx} but with $n$-bit comparators. 

\end{document}